\documentclass[letterpaper,11pt]{article}
\usepackage[utf8]{inputenc}

\usepackage{amsmath, amsthm, amssymb, thm-restate}
\usepackage{algorithmicx}
\usepackage[table,xcdraw]{xcolor}
\usepackage[vlined,linesnumbered]{algorithm2e}
\usepackage{setspace}
\usepackage{mathtools}
\usepackage[numbers]{natbib}
\usepackage{comment} % enables comment environment
\usepackage{tcolorbox} 
\usepackage{xfrac}
\usepackage{hyperref}
\usepackage{multirow}
\usepackage{caption}
\usepackage{bm}
\usepackage{newfloat}
\usepackage{enumitem}
\usepackage{dblfloatfix} 
\usepackage{wrapfig}

\usepackage{fancyhdr}
\usepackage[noabbrev]{cleveref}

\usepackage[margin=1in]{geometry}

\allowdisplaybreaks

\definecolor{mygreen}{RGB}{10,110,230}
\definecolor{myred}{RGB}{10,110,230}

\hypersetup{
     colorlinks=true,
     citecolor= mygreen,
     linkcolor= myred
}

\renewcommand{\epsilon}{\varepsilon}

\DeclareMathOperator{\E}{\ensuremath{\normalfont \textbf{E}}}

\newcommand{\hiddencomment}[1]{}

\newcommand{\mc}[1]{\ensuremath{\mathcal{#1}}}

\newcounter{protocolcounter}
\newenvironment{protocol}{
\refstepcounter{protocolcounter}
\begin{whitetbox}
\textbf{Algorithm \theprotocolcounter:}%
}{\end{whitetbox}} 
\crefname{protocolcounter}{Algorithm}{Algorithms}
\DeclareMathOperator{\poly}{poly}

\crefname{lemma}{Lemma}{Lemmas}
\crefname{theorem}{Theorem}{Theorems}
\crefname{property}{Property}{Properties}
\crefname{claim}{Claim}{Claims}
\crefname{definition}{Definition}{Definitions}
\crefname{observation}{Observation}{Observations}
\crefname{proposition}{Proposition}{Propositions}
\crefname{assumption}{Assumption}{Assumptions}
\crefname{line}{Line}{Lines}
\crefname{figure}{Figure}{Figures}
\creflabelformat{property}{(#1)#2#3}
\crefname{equation}{}{}
\crefname{section}{Section}{Sections}
\crefname{appendix}{Appendix}{Appendices}
\newtheorem{theorem}{Theorem}[section]
\newtheorem{lemma}[theorem]{Lemma}
\newtheorem{proposition}[theorem]{Proposition}
\newtheorem{corollary}[theorem]{Corollary}

\newtheorem{definition}[theorem]{Definition}
\newtheorem{claim}[theorem]{Claim}

\definecolor{mylightgray}{RGB}{230,230,230}

\algnewcommand{\IIf}[2]{\textbf{if} #1 \textbf{then} #2}
\algnewcommand{\EndIIf}{\unskip\ \algorithmicend\ \algorithmicif}

\newcommand{\restatethm}[2]{\noindent \textbf{Theorem~#1} (restated) \textbf{.} {\em #2}}

\newenvironment{graytbox}{
\par\addvspace{0.1cm}
\begin{tcolorbox}[width=\textwidth,
                  boxsep=5pt,
                  left=1pt,
                  right=1pt,
                  top=2pt,
                  bottom=2pt,
                  boxrule=0pt,
                  arc=0pt,
                  colback=mylightgray,
                  colframe=black,
                  ]%%
}{
\end{tcolorbox}
}

\newenvironment{whitetbox}{
\par\addvspace{0.1cm}
\begin{tcolorbox}[width=\textwidth,
                  boxsep=5pt,
                  left=1pt,
                  right=1pt,
                  top=2pt,
                  bottom=2pt,
                  boxrule=1pt,
                  arc=0pt,
                  colframe=black,
                  colback=white
                  ]%%
}{
\end{tcolorbox}
}

\makeatletter
\renewcommand{\paragraph}{%
  \@startsection{paragraph}{4}%
  {\z@}{10pt}{-1em}%
  {\normalfont\normalsize\bfseries}%
}
\makeatother

\makeatletter
\patchcmd{\@algocf@start}% <cmd>
  {-1.5em}% <search>
  {0pt}% <replace>
  {}{}% <success><failure>
\makeatother

\title{Streaming Edge Coloring with Asymptotically Optimal Colors}

\author{
Soheil Behnezhad \\ {\em Northeastern University} \and Mohammad Saneian  \\ {\em Northeastern University}
}

\date{}

\begin{document}

\maketitle

\thispagestyle{empty}

\begin{abstract}
    Given a graph $G$, an {\em edge-coloring} is an assignment of colors to {\em edges} of $G$ such that any two edges sharing an endpoint receive different colors. By Vizing's celebrated theorem, any graph of maximum degree $\Delta$ needs at least $\Delta$ and at most $(\Delta + 1)$ colors to be properly edge colored. In this paper, we study edge colorings in the {\em streaming} setting. The edges arrive one by one in an {\em arbitrary} order. The algorithm takes a single pass over the input and must output a solution using a much smaller space than the input size. Since the output of edge coloring is as large as its input, the assigned colors should also be reported in a streaming fashion.

    \smallskip \smallskip
    The streaming edge coloring problem has been studied in a series of works over the past few years. The main challenge is that the algorithm cannot ``remember'' all the color assignments that it returns. To ensure the validity of the solution, existing algorithms use many more colors than Vizing's bound. Namely, in $n$-vertex graphs, the state-of-the-art algorithm with $\widetilde{O}(n s)$  space\footnote{Here and throughout the paper, $\widetilde{O}(f) = O(f \cdot \poly\log n).$} requires $O(\Delta^2/s + \Delta)$ colors. Note, in particular, that for an asymptotically optimal $O(\Delta)$ coloring, this algorithm requires $\Omega(n\Delta)$ space which is  as large as the input. Whether such a coloring can be achieved with sublinear space has been left open.

    \smallskip\smallskip
    In this paper, we answer this question in the affirmative. We present a randomized algorithm that returns an asymptotically optimal $O(\Delta)$ edge coloring using $\widetilde{O}(n \sqrt{\Delta})$ space. More generally,  our algorithm returns a proper $O(\Delta^{1.5}/s + \Delta)$ edge coloring with $\widetilde{O}(n s)$ space, improving prior algorithms for the whole range of $s$.
\end{abstract}

% {
% \hypersetup{hidelinks}
% \vspace{1cm}
% \renewcommand{\baselinestretch}{0.1}
% \setcounter{tocdepth}{2}
% \tableofcontents{}
% \thispagestyle{empty}
% \clearpage
% }

\clearpage
\setcounter{page}{1}
\section{Introduction}

Motivated by its applications in processing massive graphs, the {\em graph streaming model} has  gained significant attention over the past two decades. This model assumes that the input graph is too large to be stored in memory and is presented to the algorithm in a sequence of edges. Many graph problems have been studied in this model, including maximum matching \cite{AssadiBKL-STOC23,Kapralov21,Assadi-SODA22, FeigenbaumKMSZ04,AssadiB-ICALP21}, graph connectivity \cite{AhnGM12,AssadiS-STOC23}, vertex coloring \cite{AssadiCK19,AssadiKM22,AssadiCS22,BeraCG20,ChakrabartiGS22,AssadiCGS-Arxiv}, edge coloring \cite{BehnezhadDHKS19,CharikarL21,AnsariSZ22}, cut and spectral sparsifiers \cite{AhnGM12b,KapralovLMMS14,AssadiD-SOSA21,ChenKL22,KapralovMMMNST20}, and graph clustering \cite{BehnezhadCMT22,BehnezhadCMT-SODA23,Cohen-AddadLMNP21,AhnCGMW15} among many others (this is by no means a comprehensive list). In this paper, we  continue the line of work on the {\em edge coloring} problem in the streaming model. The goal of edge coloring is to assign colors to edges of a graph such that no two adjacent edges share the same color. Since the output of edge coloring is as large as its input, a streaming algorithm cannot return it in memory. Instead, the goal is to also return the solution in a streaming fashion. Doing so, the main challenge is that the algorithm cannot ``remember'' all the reported edge colors, yet has to ensure that any two incident edges receive different colors.

Edge coloring is a fundamental problem in graph theory and has many practical applications in areas such as scheduling, communication networks, and VLSI design. By a classic result of Vizing, any graph of maximum degree $\Delta$ needs at least $\Delta$ and at most $(\Delta + 1)$ colors to be properly edge colored.\footnote{To see the lower bound, note that $\Delta$ colors are needed just to color the edges of a vertex with degree $\Delta$.} While existing algorithms for finding a $(\Delta+1)$ edge coloring are rather complicated, a $(2\Delta-1)$ coloring can be found by a simple greedy algorithm that iterates over the edges and chooses an arbitrary available color for each. Unfortunately, even this simple greedy algorithm is hard to implement in the streaming setting. Recall that the algorithm cannot remember all the assigned edge colors in memory, hence it is unclear how to verify which color is available for the next edge that arrives. 

% \sbcomment{Maybe start by vertex coloring literature citations?} 
\paragraph{Prior works:} The streaming edge coloring problem was first studied by \citet*{BehnezhadDHKS19} who gave a randomized algorithm for $O(\Delta^2)$ edge-coloring with $\widetilde{O}(n)$ space, where $n$ is the number of vertices. In a follow up work, \citet*{CharikarL21} showed that, more generally, for any parameter $s \geq 1$ there is a randomized streaming algorithm that $O(\Delta^2/s + \Delta)$ edge-colors the graph using $\widetilde{O}(ns)$ space. Later, \citet*{AnsariSZ22} obtained the same bound but using a simple and clean {\em deterministic} algorithm. Note that for an asymptotically optimal $O(\Delta)$ edge coloring, the algorithms above require $O(n \Delta)$ space. This, unfortunately, is {\em not} sublinear in the input size as any graph of maximum degree $\Delta$ has at most $O(n \Delta)$ edges. Put differently, for the case of $O(\Delta)$ coloring, no improvement over the trivial algorithm that stores the whole graph in memory and then colors it is known. This state of affairs leaves an important  question open:
\begin{quote}
{\em Does there exist a streaming algorithm with sublinear space for $O(\Delta)$ edge coloring?}
\end{quote}

\paragraph{Our contribution:} In this paper, we answer the question above in the affiramtive. Our main result is the following algorithm:

\begin{graytbox}
\begin{theorem}\label{thm:main}
    For any $s \geq 1$, there is a randomized streaming algorithm that  with high probability reports a $O(\Delta^{1.5}/s + \Delta)$ edge-coloring  under arbitrary edge arrivals using $\widetilde{O}(n s)$ space.
\end{theorem}
\end{graytbox}

Setting $s = \sqrt{\Delta}$, we obtain the following corollary, answering the question above:

\begin{graytbox}
\begin{corollary}\label{cor:asym-optimal}
There is a randomized streaming algorithm that with high probability reports a $O(\Delta)$ edge-coloring  under arbitrary edge arrivals using $\widetilde{O}(n \sqrt{\Delta})$ space.
\end{corollary}
\end{graytbox}

We show that the space-complexity can be further improved to $\widetilde{O}(n)$ under arbitrary {\em vertex} arrivals. Here instead of edges arriving in an arbitrary order, the vertices of the graph arrive one by one and when a vertex $v$ arrives, all of its edges to the previous vertices are revealed.

\newcommand{\thmvertexarrival}[0]{
There is a randomized streaming algorithm that  $O(\Delta)$ edge colors the graph under arbitrary vertex arrivals using $\widetilde{O}(n)$ space.
}

\begin{graytbox}
\begin{theorem}\label{thm:vertex-arrival}
    \thmvertexarrival{}
\end{theorem}
\end{graytbox}

\subsection{Further Related Work}

Streaming algorithms for edge coloring have also been studied under the extra assumption that the edges arrive in a random order. \citet*{BehnezhadDHKS19} showed that there is a single pass $\widetilde{O}(n)$ space algorithm that obtains a $5.44\Delta$ coloring. \citet*{CharikarL21} later showed the number of colors can in fact be improved to $(1+o(1))\Delta$ under random arrivals while keeping the memory $\widetilde{O}(n)$. Both of these algorithms rely heavily on the random-arrival assumption and do not have any implications for adversarial edge arrivals, which is the focus of this paper.

Online edge-coloring is another related problem. In this problem, the algorithm has no space constraints, but edges arrive one by one and each edge has to be colored upon arrival irrevocably. Note that the greedy algorithm can easily be implemented in the online setting; therefore much of the research has been focused on whether the greedy bound can be improved. For low-degree graphs, \citet*{Bar-NoyMN92} showed that the greedy algorithm is optimal for online edge coloring. However, there has been several improvements over the past few years for graphs of degree at least $\omega(\log n)$ starting from the work of \citet*{CohenPW19}; see \cite{CohenPW19,SaberiW21,BhattacharyaGW21,KulkarniLSST22} and the references therein. The best current bound for general edge arrivals is a beautiful $\frac{e}{e-1}\Delta$ coloring algorithm of \citet*{KulkarniLSST22}. Whether the number of colors can be improved to $(1+o(1))\Delta$ for $\Delta = \omega(\log n)$ under arbitrary edge arrivals remains an important open problem.

\paragraph{Independent work:} In a concurrent and independent work, \citet*{concurrent} have also considered edge colorings in the streaming setting. In particular, they also achieve the same bounds as in our \Cref{cor:asym-optimal} and \Cref{thm:vertex-arrival}. A smooth color/space trade-off is also presented in \cite{concurrent}. In particular, for any $1\leq s \leq \sqrt{\Delta}$, they achieve an $\widetilde{O}(ns)$ space algorithm that colors the graph using $\widetilde{O}(\Delta^2/s^2)$ colors.\footnote{In their paper this is equivalently stated as an $O(\Delta t)$ coloring using $\widetilde{O}(n \sqrt{\Delta/t})$ space.} This can be compared to our \cref{thm:main} which uses $O(\Delta^{1.5}/s)$ colors with $\widetilde{O}(ns)$ space. We note that $\Delta^2/s^2 \geq \Delta^{1.5}/s$ for the whole range of $s$. So our space/color trade-off never uses more colors. But for small $s$ our algorithm uses fewer colors. For instance, in the extreme case of $\widetilde{O}(n)$ space algorithms (i.e. with $s = 1$), our algorithm uses $O(\Delta^{1.5})$ colors whereas that of \cite{concurrent} requires $O(\Delta^2)$ colors.  Finally, we note that while we only considered simple graphs and randomized algorithms in this work, \cite{concurrent} also extend their algorithms to multi-graphs and derandomize them.

\clearpage 

% The two works seem to use different techniques. 

% In particular, they achieve an $O(\Delta t)$ coloring using $\widetilde{O}(n \sqrt{\Delta/t})$ space. We note that this trade-off is never better than the color/space trade-off of our \cref{thm:main}. 

% A smooth color/space trade-off is also presented in \cite{concurrent}. In particular, they achieve an $O(\Delta t)$ coloring using $\widetilde{O}(n \sqrt{\Delta/t})$ space. Equivalently, for any $1 \leq s \leq \sqrt{\Delta}$, this implies an $O(\Delta^2/s^2)$ coloring using $\widetilde{O}(ns)$ space. This has to be compared to our  \cref{thm:main} which requires $O(\Delta^{1.5}/s)$ colors with $\widetilde{O}(ns)$ space when $1 \leq s \leq \sqrt{\Delta}$.  

% They also show that many of their algorithms can be de-randomized whereas we only consider randomized algorithms in this work. The space/color trade-off of \cref{thm:main} is unique to our paper.

\subsection{Preliminaries}

Unless otherwise stated, we use $G=(V, E)$ to denote the input graph. We use $n := |V|$ and $m := |E|$ to respectively denote the number of vertices and edges in $G$. We use $\Delta$ to denote the maximum degree of the graph $G$. 
For any integer $k$, we use $[k]$ to denote the set $\{1, \ldots, k\}$. 

\paragraph{The graph streaming model:}
In the standard graph streaming model, edges of an arbitrary graph $G$ arrive one by one in an arbitrary order. The algorithm has a space much smaller than the total number of edges, can take few---preferrably just one---pass over the input, and should return the output. In the case of edge coloring, the output size is as large as the input, making it impossible to store the entire output and return it all at once. Therefore, we allow the algorithm to output the solution in a streaming manner as well. This model is also referred to as the ``W-streaming'' model in the literature \cite{Glazik2017, Demetrescu10}. All of the algorithms presented in this paper take only a single pass over the input. We measure the space in the number of words, each consisting of $\Theta(\log n)$ bits. 

Our algorithms for the general edge arrival model build on algorithms that we develop for two more restrictive arrival models of {\em general vertex arrivals} and {\em one-sided vertex arrivals in bipartite graphs}. We present the definition of these standard models below.

\begin{definition}[vertex arrival model]\label{def:vertex-arrival}
In this model, vertices of the input graph $G$ arrive one by one according to some arbitrary permutation $\pi$. Upon arrival of a vertex, all of its edges to previous vertices in the permutation $\pi$ arrive.
\end{definition}

\begin{definition}[one-sided vertex arrivals in bipartite graphs]\label{def:one-sided-vertex-arrival}
In this model, the input graph $G$ is assumed to be bipartite with vertex sets $U$ and $V$. The ``offline'' vertices in $V$ are present from the beginning, but the ``online'' vertices in $U$ arrive one by one in an arbitrary order. Every time an online vertex $u$ arrives, all of its edges to the offline vertices $V$ are revealed.
\end{definition}

In our proofs, we use the following standard variant of the Chernoff bound.

\begin{proposition}[Chernoff bound] \label{prep:chernoff}
Let $X_1, ..., X_n$ be independent random variables in $[0, 1]$. Let $X = \sum_{i=1}^n X_i$ and $\mu = \E[X]$. Then for all $\delta \geq 0$ and $\mu' \geq \mu$, $
    \Pr[X \geq (1 + \delta)\mu'] \leq \exp \left( -\frac{\delta^2}{2 + \delta} \mu' \right).
    $
\end{proposition}

\section{Overview of Techniques}

In this section, we give an informal high-level overview of our algorithms. 

As we discussed, the main challenge in solving the streaming edge-coloring problem is that the algorithm cannot ``remember'' all the colors that we assign to the edges as this takes too much space. This turns out to be a challenge particularly when the degrees evolve unevenly. To convey the key intuitions in this section, let us first focus on the one-sided vertex arrival model in bipartite graphs (\cref{def:one-sided-vertex-arrival}). We note that even in this restricted arrival model, the best known algorithm from the literature remains to be those of \cite{BehnezhadDHKS19,CharikarL21,AnsariSZ22} which require $O(\Delta^2)$ colors with $O(n)$ space. Here, we describe how this can be improved to an asymptotically optimal $O(\Delta)$ coloring with only $\widetilde{O}(n)$ space. Our final algorithm of \cref{thm:main} in the more general edge-arrival model builds on this vertex-arrival algorithm.

Since all edges of an online vertex arrive at the same time, it is not hard to ensure they  receive different colors. What is challenging is to do so while ensuring that all edges of an {\em offline} vertex receive different colors too. Towards this, we first describe an algorithm that uses $O(\Delta \log n)$ colors, $O(n \Delta)$ ``pre-processing space,'' and $O(n)$ working space. We then show how this can be turned into an $O(\Delta)$ coloring algorithm that uses $\widetilde{O}(n)$ space overall.

\paragraph{An algorithm with $O(\Delta \log n)$ colors but large space:} Let $K = \Theta(\Delta \log n)$ be the number of colors we use. In the pre-processing step (i.e., before seeing any edges of the graph), for every offline vertex $v \in V$, we store a random permutation $\pi_v$ of the colors $\{1, \ldots, K\}$ in memory which overall takes $\widetilde{O}(n \Delta)$ pre-processing space. Now suppose that the first online vertex $u$ arrives. For each of its offline neighbors $v_i$, we consider the first random color $\pi_{v_i}(1)$ of $v_i$. Since the random permutations of offline vertices are independent, when $\Delta \geq 2$ we get the color chosen for $v_i$ is different from $\pi_{v_j}(1)$ for all other neighbors $v_j$ of $u$ with probability  $(1-\frac{1}{K})^{\deg(u)} \geq (1-\frac{1}{\Delta})^\Delta \geq 0.25$. Note that $(1-\frac{1}{\Delta})^\Delta$ is a monotonically increasing and for $\Delta = 2$ the value is $0.25$. If this happens, we assign color $\pi_{v_i}(1)$ to edge $(u, v_i)$. If the first color of $v_i$ is not unique, we discard it and reveal the next color $\pi_{v_i}(2)$ of $v_i$. Every round of this process successfully colors a constant fraction of the remaining uncolored edges of $u$. Thus it takes $O(\log n)$ rounds to color all edges of $u$ w.h.p. On the other hand, since we have reserved $K = \Theta(\Delta \log n)$ colors for each offline vertex, we can afford to reveal up to $\Theta(\log n)$ colors for each of their edges, hence fully coloring the graph w.h.p. To implement this algorithm, we only need to maintain a counter $d_v$ for every offline vertex $v$ on how many colors of its random permutation we have revealed so far, which can be done with $O(n)$ total working memory.

\paragraph{Reducing space:} To get rid of the huge pre-processing space of $\widetilde{O}(n \Delta)$ in the previous algorithm, we limit the amount of randomness needed. To do so, instead of choosing a fully random permutation of $\{1, \ldots, K\}$ for every offline vertex $v$, which requires $\Theta(n K) = \widetilde{\Theta}(n \Delta)$ space, we just pick a random number $r_v \in [K]$ and use the randomly {\em shifted} permutation $(r_v, r_v + 1, \ldots, K, 1, \ldots, r_v - 1)$. The advantage of doing so is that this shifted permutation can be stored using $O(1)$ space, by just storing the random number $r_v$ in memory. Its downside is that lack of independence breaks the analysis above. For example, if $\pi_{v_i}(1) = \pi_{v_j}(1)$ then we also have $\pi_{v_i}(2) = \pi_{v_j}(2)$ and so cannot argue that every round colors a constant fraction of the edges of the online vertex $u$. To fix this, we take $t = \Theta(\log n)$ random numbers $r^1_v, \ldots, r^t_v$ for each offline vertex $v$ and considering the $t$ shifted permutations
$$
[r^1_v, \ldots, \Delta, 1, \ldots, r^1_v - 1], \ldots, [r^t_v, \ldots, \Delta, 1, \ldots, r^t_v - 1].
$$
Now if the first color choice of $v_i$ and $v_j$ according to the first shifted permutation are the same, we consider the second shifted permutations, then the thirds, and so on and so forth. This gets rid of the dependence between the colors proposed for the edges of an online vertex, but not the edges of an offline vertex. Luckily, the latter is not needed for the analysis to go through and we can implement this algorithm with $\widetilde{O}(n)$ space overall.
% \mscomment{reviwer 1 has a comment for here which i don't think needs editing. We have clearly explained the main idea and discussed multiple online vertices too! }

\paragraph{Reducing colors via $k$-out sampling:} The algorithm above is greedy in that we first reveal a random color for each edge of the online vertex $u$, greedily color those whose proposed colors are unique, then reveal the next batch of proposals. Instead of this greedy algorithm which takes up to $O(\log n)$ rounds, inevitably stretching the number of colors to $O(\Delta \log n)$, we first draw 3 random colors $x^1_i, x^2_i, x^3_i$ for each edge $e_i = (u, v_i)$ of the online vertex $u$ and consider all these 3 colors {\em at the same time}. We show that with probability $1-1/\poly(\Delta)$, it is possible to pick one of the colors $x^1_i, x^2_i, x^3_i$ for each edge $e_i$ of $u$ such that all edges of $u$ receive different colors. Our proof of this theorem builds on a new lemma (\cref{lem:kout}) that we prove on the existence of perfect matchings in a {\em one-sided random $k$-out model}. This lemma, which might be of independent interest, says that if we have a random bipartite graph with vertex sets $V$ and $U$, and every vertex in $V$ is made adjacent to exactly $k$ vertices in $U$ chosen uniformly, then $G$ has a perfect matching provided that $k \geq 3$ and $|U| > e|V|$. While $k$-out sampling has been used recently in obtaining more efficient algorithms for distributed graph connectivity \cite{HolmKTZZ19} and minimum cuts \cite{Ghaffari0T20,AssadiD21}, this is to our knowledge its first application in graph coloring.

\paragraph{From vertex-arrivals to edge-arrivals:} We now overview how we go from the vertex arrival model to the more general edge arrival model. To convey the key intuitions, we focus only on the simpler special case of \cref{thm:main} where $s = 1$. That is,  an algorithm that uses $O(n)$ space and returns an $O(\Delta^{1.5})$ edge coloring. Our first step is to generalize the algorithm above to work in the non-bipartite vertex arrival model (\cref{def:vertex-arrival}). This step, in fact, follows more or less from a known random bipartization technique of the literature that we present in \cref{apx:bipartization}. Our next, more challenging, step is to generalize the algorithm to a {\em batch} arrival setting, where at each step instead of all of the edges of a vertex, we see some $\Theta(\sqrt{\Delta})$ edges of it (without having any guarantee about when the next batch of this vertex arrives). Once we achieve this more general algorithm, we run an instance $\mc{A}$ of it. To feed our edges to this batch arrival algorithm $\mc{A}$, we keep storing edges in a set $H$. Whenever the number of edges in $H$ reaches $n$, we look at the vertex $v$ with the largest degree in $H$. If $\deg_H(v) \geq \sqrt{\Delta}$, we feed $\Theta(\sqrt{\Delta})$ of these edges of vertex $v$ as the next batch to algorithm $\mc{A}$ and remove them from $H$. Otherwise, the maximum degree in $H$ is less than $\sqrt{\Delta}$; in this case, we edge color all edges in $H$ greedily using $O(\sqrt{\Delta})$ colors and remove them all from $H$. Every time that we color $H$ greedily, we color $n$ edges of the graph. Therefore this happens at most $m/n = O(\Delta)$ times, requiring a total of $O(\Delta^{1.5})$ colors. The proof of \cref{thm:main} for larger values of $s$ requires a more involved white-box application of the vertex-arrival algorithm; see \cref{sec:edge-arrival} for the details of the algorithm.

\section{Streaming Edge Coloring Under Vertex Arrivals}\label{sec:vertex-arrival}

In this section, we start by proving \cref{thm:vertex-arrival}, restated below, for the streaming edge-coloring under vertex arrivals. Our algorithm in the more general edge-arrival model builds on the vertex-arrival algorithm that we describe in this section.

\restatethm{\ref{thm:vertex-arrival}}{\thmvertexarrival{}}

\subsection{Basic Reductions}

We start with two basic reductions that essentially reduce the edge coloring problem in general graphs under vertex arrivals to the same problem in bipartite graphs under one-sided vertex arrivals.

The following lemma asserts that instead of general graphs, we can focus on bipartite graphs. The idea is to randomly partition the vertex set $V$ into two sets $A$ and $B$, edge color the bipartite subgraph of $G$ between $A$ and $B$, and recurse on the induced subgraphs $G[A]$ and $G[B]$. Since the idea is standard and follows from known reductions, we present its proof in \hyperref[appendix:a]{Appendix A}.

\begin{lemma}\label{lem:bipartization}
Suppose there is a streaming algorithm that edge colors any $n$-vertex bipartite graph of maximum degree $\Delta$ with $f(\Delta) \geq \Delta$ colors using $s(n, \Delta)$ words of space. Then there is a streaming algorithm that edge colors any general (i.e., not necessarily bipartite) $n$-vertex graph of maximum degree $\Delta$ using
$$
f(\Delta) + f(\Delta/2) + f(\Delta/4) + ... + 2f(1)
$$
colors and $\widetilde{O}(n + s(n, \Delta))$ space. If the algorithm for bipartite graphs works under vertex arrivals (resp. edge arrivals), then the algorithm for general graphs also works under vertex arrivals (resp. edge arrivals).
\end{lemma}

We emphasize that to solve general graphs under vertex arrivals, we have to ensure that the bipartite algorithm that we feed into \cref{lem:bipartization} works under two-sided vertex arrivals. The next reduction shows that any one-sided vertex arrival algorithm also solves two-sided vertex arrivals by losing a factor of 2 in the number of colors.

\begin{claim}\label{cl:twosided-to-onesided}
    Suppose there is a streaming algorithm that edge colors any $n$-vertex bipartite graph of maximum degree $\Delta$ under one-sided vertex arrivals using $f(\Delta)$ colors. Then this can be turned into an algorithm that $2f(\Delta)$ colors the bipartite graph under two-sided vertex arrivals using asymptotically the same space.
\end{claim}
\begin{proof}
Let $G=(V, U, E)$ be the bipartite graph under two-sided vertex arrivals. We partition its edges into two subgraphs $G_U$ and $G_V$. If a vertex in $V$ arrives, we add its edges to $G_V$, and if a vertex in $U$ arrives, we add its edges to $G_U$. This way, $G_V$ (resp. $G_U$) will be a bipartite graph under one-sided vertex arrivals with the online part being vertex set $V$ (resp. $U$). Since both $G_U$ and $G_V$ will have maximum degree upper bounded by $\Delta$ (as they are subgraphs of $G$), we can run one instance of \cref{lem:onesided} on $G_U$ and one on $G_V$ in parallel using disjoint colors. This only multiplies the number of colors by two and the space by two, finishing the proof.
\end{proof}

\subsection{The Algorithm}

The reductions of the previous section show that instead of general graphs, we can focus on bipartite graphs under one-sided vertex arrivals. The following lemma is our main contribution in the rest of \cref{sec:vertex-arrival}.

\begin{lemma}\label{lem:onesided}
    % There is a streaming algorithm that edge colors any bipartite graph of maximum degree $\Delta$ using $(\frac{e}{e-1}+\epsilon \sim 1.58)\Delta$ colors under one-sided vertex arrivals with $O(n\log n / \epsilon^2)$ space.
    There is a streaming algorithm that edge colors any bipartite graph of maximum degree $\Delta$, using $O(\Delta)$ colors under one-sided vertex arrivals and uses $O(n)$ space w.h.p. (The coloring is always valid and the space-bound holds with probability $1 - 1/\poly(n)$.) 
\end{lemma}

Let us first see how \cref{lem:onesided} proves \cref{thm:vertex-arrival} using the reductions of the previous section.

\begin{proof}[Proof of \cref{thm:vertex-arrival}]
    First, we feed the algorithm of \cref{lem:onesided} to \cref{cl:twosided-to-onesided} to obtain a streaming algorithm that for some constant $C$, $C(\Delta)$ colors a bipartite graph under two-sided vertex arrivals with $O(n)$ space. We then plug in this algorithm into \cref{lem:bipartization} to obtain an algorithm for general graphs. The algorithm uses $\widetilde{O}(n)$ space and the number of its colors by \cref{lem:bipartization} is
    $$
        C\Delta + C(\Delta/2) + C(\Delta/4) + \ldots + C \leq 2C\Delta = O(\Delta).\qedhere
    $$
\end{proof}

So it only remains to prove \cref{lem:onesided}. In this section, we present the algorithm formalized below as  \cref{alg:one}. We analyze the space complexity of the algorithm in \cref{sec:space} and analyze its correctness and the number of used by it in \cref{sec:correctness}. Finally, we show in \cref{sec:batch} that the algorithm can be generalized to a batch arrival model as well.

\begin{figure}[h]
\begin{protocol}\label{alg:one}
    Streaming edge coloring for one-sided vertex arrivals in bipartite graphs.

    \smallskip\smallskip
    \textbf{Parameter:} $c := 2.72$
    \begin{enumerate}
        \item For each offline vertex $v$ draw $3$ {\em distinct} random numbers $r_v^1, r_v^2$ and $r_v^3$ from $[c\Delta]$ uniformly and store them in memory.\footnote{For brevity, here and in the rest of the analysis we assume that $c \Delta$ is an integer. If it is not, $c\Delta$ must be replaced by $\lceil c\Delta \rceil$.} 
        \item Additionally, for each offline vertex $v$, we store a counter  $\deg_v$ in memory to keep track of its degree as the edges arrive.
        \item Upon arrival of an online vertex $u$:
        \begin{enumerate}
            %[ref=\theenumi-(\theenumii)]
            \item For each edge $e_i = (u, v_i)$ consider the following three colors $x^1_i$, $x^2_i$, $x^3_i$:
            \begin{flalign*}
                x^1_i &:= \big((r^1_{v_i} + \deg_{v_i}) \bmod c \Delta\big),\\
                x^2_i &:= \big((r^2_{v_i} + \deg_{v_i}) \bmod c \Delta\big) + c\Delta,\\
                x^3_i &:= \big((r^3_{v_i} + \deg_{v_i}) \bmod c \Delta \big) + 2c\Delta.
            \end{flalign*}\label{alg1_colordefine}
            \vspace{-0.6cm}
            \item If there is an assignment of colors to edges of $u$ such that each edge $e_i$ receives a color from $\{x^1_i, x^2_i, x^3_i\}$ and all the edges of $u$ receive different colors, then we assign these colors, stream them out, and remove edges of $u$ from memory. Otherwise, we add all edges of $u$ to set $S$ which we store in memory. \label{alg1_findmatching}

            (In \cref{cl:notmanyfails} we show $u$'s edges are successfully colored with probability $1-O(1/\Delta^5)$.)
        \end{enumerate}
        \item It only remains to color the edges in $S$. Note that since $S$ is a subgraph of $G$, its maximum degree is no larger than $\Delta$. Thus we can edge-color it using $\Delta$ fresh colors via existing offline edge-coloring algorithms for bipartite graphs; see e.g. \cite{gabow1985algorithms}. \label{alg1_end}
    \end{enumerate}
\end{protocol}
\end{figure}

\subsection{Space Complexity}\label{sec:space}

In \cref{alg:one}, for every vertex, we keep its degree and three random numbers which together can all be stored with $O(n)$ words. The only non-trivial part of the space that we need to bound is the size of the set $S$ of the edges that we store in memory and color at the end. Our main result in this section is to show that the set $S$ has size $O(n)$ w.h.p. We start by bounding the expected size of $S$ in \cref{cl:notmanyfails} using a connection to $k$-out subgraphs (formalized in \cref{lem:kout}). We then prove a high probability bound on the size of $S$ in \cref{cl:space-highprob}.

% The space complexity of \cref{alg:one} except for the edges in $S$ is deterministically at most $O(n)$. 

% Every part of the space we are using except for the edges kept in $S$ at step 3.b uses $O(n)$ space. Since for every vertex, we keep its degree and three random numbers which together can all be stored with $O(n)$ words. Hence, we need to provide a bound on $|S|$. We focus our attention in this subsection mainly for this purpose. The following claim is the key to bounding the memory our algorithm uses. 

\begin{claim}\label{cl:notmanyfails}
Take any online vertex $u$. The probability that we store the edges of $u$ in Step~\ref{alg1_findmatching}  of \cref{alg:one} is at most $O(1/\Delta^5)$. This, by linearity of expectation, implies that
$$
    \E|S| \leq O(n / \Delta^4).
$$
\end{claim}

\begin{proof}
Our main idea to prove this claim is to use the maximum matching problem on an appropriately defined ``color graph'' defined below. When an online vertex $u$ arrives in the input we construct the bipartite graph $H_u$ in the following way.

\begin{definition}[the color graph]\label{def:colorgraph}

For any online vertex $u$, we define the color graph $H_u = (X, Y, E)$ as follows. The set $X$ corresponds to the edges of $u$, i.e., for each edge $(u, v_i)$, we have a vertex corresponding to $v_i$ in $X$. For each color in the range $[c \Delta]$
 we have a corresponding vertex in set $Y$. Each vertex in $X$ has exactly three edges to $Y$ which are to the three colors:
\begin{flalign*}
                y^1_i &:= \big((r^1_{v_i} + \deg_{v_i}) \bmod c \Delta\big),\\
                y^2_i &:= \big((r^2_{v_i} + \deg_{v_i}) \bmod c \Delta\big),\\
                y^3_i &:= \big((r^3_{v_i} + \deg_{v_i}) \bmod c \Delta \big),
            \end{flalign*}
where $v_i$ is any offline vertex adjacent to $u$. Note that since $r^1_{v_i}, r^2_{v_i}, r^3_{v_i}$ are distinct in \cref{alg:one}, the three colors $y^1_i, y^2_i, y^3_i$ will be distinct as well.
\end{definition}

We first describe that a perfect matching in $H_u$ corresponds to a valid edge-coloring of all edges of $u$. Let $M$ be a perfect matching in $H_u$. Suppose that the node corresponding to edge $e_i$ is matched in $M$ to color $y_i^j$ (which is distinct for all $e_i$'s). In this case, we assign color $x_i^j$ to $e_i$ where $x^i_j$ is as defined in \cref{alg:one}. To see why all edges of $u$ receive different colors note that 
$$y_i^j = x_i^j \bmod c\Delta.$$ 
In addition, all matched $y_i^j$ are unique because $M$ is a matching. Therefore, if we find a perfect matching in the color graph $H_u$ we can easily color edges incident to $u$ using this matching.  

Thus, it remains to show that $H_u$ does indeed have a perfect matching (with large enough probability). We do so using the following \cref{lem:kout} on existence of perfect matchings in a so-called $k$-out model. \cref{lem:kout} proves that if each vertex of one side of a bipartite graph is made adjacent to 3 random vertices on the other side, then the graph has a perfect matching provided that the other side is sufficiently large.  We note that by a classic result of \citet*{Frieze86} from 1986, if vertices in {\em both} sides of the bipartite graph pick $k$ random edges, then the graph has a perfect matching w.h.p. if and only if $k \geq 2$. However, this does not hold when only one side of the graph pick random edges, which is the focus of the following \cref{lem:kout}.

% We note that equivalent of \cref{lem:kout} has been proved before in the model where both vertex sets in the graph pick $k$ edges  

\begin{lemma}\label{lem:kout}
    Consider a random bipartite graph $G$ with vertex sets $V$ and $U$ where each vertex $v \in V$ is adjacent to exactly 3 distinct vertices in $U$ picked uniformly (from all subsets of size 3 of $U$) and independently from the choice of the rest of vertices in $V$. If $|V| \leq n$ and $|U| = cn$ for fixed $c > e$, the graph $G$ has a perfect matching with probability at least $1-O(1/n^5)$.

\end{lemma}

\begin{proof}
    First, we argue that we can w.l.o.g. assume that $|V| = n$. Since if we have $|V| < n$ we can add some dummy vertices to side $V$ and draw random neighbors for them to $U$. Any perfect matching in this new graph gives us a perfect matching in the original graph by removing its dummy edges. Therefore the probability of finding a perfect matching in this graph is no smaller than the case with $|V| = n$. 
    
    We let $c = (1+\epsilon)e$ for some fixed $\epsilon > 0$. We assume that $n$ is larger than any needed constant (possibly a function of $\epsilon$). Note that if $n$ is fixed, then by adjusting the constants we can ensure $1-O(1/n^5) = 0$, making the probablistic statement of the lemma trivial.

    By Hall's theorem, $G$ has a perfect matching iff for every subset $S \subseteq V$ it holds that $|N_G(S)| \geq |S|$. This condition deterministically holds if $|S| \leq 3$ as every vertex in $V$ has degree exactly 3. Therefore it suffices to show that it holds with high probability for all sets of size  at least 4. Let us fix $4 \leq k \leq n$. The probability that all the edges of a vertex $v \in V$ go to a fixed subset $U$ of size at most $k-1$ in $U$ is exactly $\binom{k-1}{3}/ \binom{|U|}{3}$. Therefore, we have
    \begin{flalign*}
        \Pr[\text{exists } S \subseteq V : |S| = k \text{ and } |N(S)| < k] &\leq \binom{n}{k} \cdot  \binom{|U|}{k-1} \cdot \left(\frac{\binom{k-1}{3}}{\binom{|U|}{3}}\right)^k\\
        &\leq \left( \frac{n e}{k}\right)^k \cdot \left( \frac{cne}{k-1}\right)^{k-1} \cdot \left(\frac{(k-1)(k-2)(k-3)}{cn (cn-1)(cn-2)}\right)^{k} \tag{since $\binom{n}{k} \leq \left(\frac{ne}{k}\right)^k$ for all $n$ and $k$}\\
        &\leq \left( \frac{n e}{k}\right)^k \cdot \left( \frac{cne}{k-1}\right)^{k-1} \cdot \left(\frac{k-1}{cn/(1+0.1\epsilon)}\right)^{3k}\tag{Here we use $cn/(1+0.1\epsilon) \leq cn-2$ which holds for $n$ larger than some constant.}\\
        &= \frac{(1+0.1\epsilon)^{3k} e^{2k-1} (k-1)^{2k+1}}{c^{2k+1} \cdot k^k \cdot n^{k+1}}\\
        &\leq \left(\frac{(1+0.5\epsilon)e}{c}\right)^{2k-1} \cdot \left( \frac{k}{n}\right)^{k+1}\tag{Here we use $(1 + 0.1 \epsilon)^2 \leq 1 + 0.5 \epsilon$ for small enough $\epsilon$}
        \\
        &\leq \left(\frac{1+0.5\epsilon}{1+\epsilon}\right)^{2k-1} \left( \frac{k}{n}\right)^{k+1} \tag{Since $c =(1+\epsilon)e.$}.
    \end{flalign*}

    To finish the proof, note that as discussed, for $k \leq 3$ the Hall's guarantee holds deterministically. For $k = 4$, the inequality above is upper bounded by $(k/n)^{k+1} = O(1/n^5)$ and thus all subsets of size $k=4$ in $V$ have at least 4 neighbors in $U$ with probability at least $1-O(1/n^5)$. For each choice of $5 \leq k \leq n^{0.1}$, the inequality above is upper bounded by $(k/n)^{k+1} \leq n^{-0.9 \times 6} = n^{-5.4}$; a union bound over all $O(n^{0.1})$ such choices of $k$ gives that with probability at least $1-O(n^{-5.3})$, any set of size $5 \leq k \leq n^{0.1}$ in $V$ has at least $k$ neighbors in $U$ as well. Finally, for the case where $n^{0.1} < k \leq n$, the inequality above is upper bounded by
    $$
        \left(\frac{1+0.5\epsilon}{1+\epsilon}\right)^{2k-1} \left( \frac{k}{n}\right)^{k+1} \leq \left(\frac{1+0.5\epsilon}{1+\epsilon}\right)^{2k-1} = 2^{-O_\epsilon(k)} = 2^{-O_\epsilon(n^{0.1})} \ll O(1/n^6).
    $$
    Thus, again by a union bound over all such choices of $k$, we get that with probability at least $1-O(1/n^{5})$, any set of size $n^{0.1} < k \leq n$ in $V$ has at least $k$ neighbors in $U$ as well. Putting together all of these cases of $k$, we get that with probability at least $1-O(1/n^5)$, Hall's condition holds and $G$ has a perfect matching.
\end{proof}

Now observe that the color graph $H_u$ meets the conditions of \cref{lem:kout}. There are at most $\Delta$ edges incident to $u$ so $|X| \leq \Delta$. On the other hand $|Y| \geq c \Delta$. Moreover, for each edge $e_i = (u, v_i)$ its corresponding vertex in $X$ is made adjacent to the vertices corresponding to colors $y_i^1, y_i^2, y_i^3$ as specified in \cref{def:colorgraph}. Since the $r_{v_i}^j$'s are distinct and uniform, each vertex in $X$ is made adjacent to 3 distinct uniform vertices in $Y$ in graph $H_u$ and so we can apply \cref{lem:kout}. This implies that $H_u$ has a perfect matching with probability at least $1 - O(1/\Delta^5)$. This means that we store the edges of each online vertex in $S$ with probability at most $O(1/\Delta^5)$, concluding the proof of \cref{cl:notmanyfails}.
\end{proof}

It is worth noting that to find the maximum matching in $H_u$ we use the algorithm of \cite{4569670} that has space complexity of $O(|V|+|E|)$ which is $O(\Delta)$ here. We can delete this part of the memory before seeing the next online vertex.

While \cref{cl:notmanyfails} bounds the expected space of \cref{alg:one}, our next \cref{cl:space-highprob} bounds the space by $O(n)$ w.h.p.

\begin{claim}\label{cl:space-highprob}
\cref{alg:one} uses $O(n)$ space with probability $1-1/\poly(n)$.
\end{claim}
\begin{proof}
For an online vertex $u$, define $X_u$ to be the indicator random variable where $X_u = 1$ iff we store edges of $u$ in $S$. Additionally, define $X = \sum_{u \in U} X_u$.
We know that
$$
\Pr[|S| > O(n)] \leq \Pr[\Delta \cdot X > O(n)] = \Pr[X > O(n/\Delta)].
$$

We know that $\E[X] \leq n \cdot O(1/\Delta^5)$ by \cref{cl:notmanyfails}. However, since the $X_u$'s  are {\em not} independent we can not use Chernoff bound to provide  a concentration bound on their sum. To see why, take two online vertices $u_1, u_2$ that share the same set of offline neighbors and come one after another in the input. If $X_{u_1} = 1$ then also $X_{u_2} = 1$ since the colors used for $H_{u_2}$ are exactly the same as the colors used for $H_{u_1}$ except they are shifted by one.

To bound $X$ our plan is to divide all $X_u$'s into groups such that in each group all the random variables are independent and then apply the Chernoff bound on each group separately. We know that if two online vertices $u$ and $w$ do not share any offline vertex as their neighbor then $X_u$ and $X_w$ are independent. This is because the colors specified in Step~\ref{alg1_colordefine} of \cref{alg:one} are all a function of the randomness on the set of neighbors an online vertex has. So if these sets are disjoint then $X_u$ and $X_w$ are independent. 

Let us define the dependency graph between these variables as follows: The vertex set  of this dependency garph is $\{X_u : u \in U\}$ and there is an edge between $X_u, X_w$ with $u, w \in U$ iff $u$ and $w$ share an offline neighbor. The maximum degree in this graph is at most $\Delta^2$ since each vertex $u \in U$ has at most $\Delta$ offline neighbors that each has at most $\Delta$ online neighbors.

Since any graph of maximum degree $\Delta'$ can be {\em vertex} colored using at most $\Delta'+1$ colors via a simple sequential greedy algorithm, the dependency graph can be vertex colored via at most $\Delta^2+1$ colors. Doing so, note that vertices in each color class become independent random variables by definition of the dependency graph. Therefore,  we can apply the Chernoff bound on the sum in each color class. With a slight abuse of notation, let $X_1, ..., X_t$ be a group of variables in a color class. Define $\mu = \E[X_1+X_2+...+X_t]$. Since $\E[X_i] = O(1/\Delta^5)$  (by \cref{cl:notmanyfails}) by linearity of expectation we get  that for some constant $C$, 
\begin{equation}\label{eq:linearity-htu8}
\mu = O(t /\Delta^5).  
\end{equation}
Applying the Chernoff bound we get that for any $\delta > 0$ and $\mu' \geq \mu$,
\begin{equation}\label{eq:cher}
\Pr[X_1+X_2+...+X_t \geq (1 + \delta)\mu'] \leq \exp\left( \frac{-\delta^2}{2 + \delta} \mu' \right).
\end{equation}
Let us call this color class {\em small} if $t \leq n/\Delta^3$ and {\em large} otherwise. For small groups, we do not need to prove any concentration bound, since a total of $(\Delta^2 + 1) (n/\Delta^3) = O(n/\Delta)$ vertices belong to small groups and each stores $\Delta$ edges in $S$, giving a deterministic upper bound of $O(n)$ on the number of such edges in $S$. It thus remains to analyze large groups only.

To deal with large groups we consider two cases for $\Delta$.

\paragraph{Case 1: $\Delta = O((n/\log n)^{1/4})$.} Letting $\delta = \Delta^4$ in \cref{eq:cher}, we get:
\begin{flalign*}
\Pr\left[\sum_{i=1}^t X_i \geq (1+\Delta^4)\mu\right] &\leq 
\exp\left(\frac{-\Delta^8}{2 + \Delta^4} \mu \right).
\end{flalign*}

Let $\mu' = C \cdot t/\Delta^5$. Note that $\mu \leq C \cdot t/\Delta^5 = \mu'$

and $t \geq \frac{n}{\Delta^3}$ since we are only considering large groups we get,

$$
\exp\left(\frac{-\Delta^8}{2 + \Delta^4} \mu' \right) \leq \exp\left(
\frac{-\Delta^8}{2 + \Delta^4} \cdot \frac{Cn}{\Delta^8}
\right) = \exp \left( \frac{-Cn}{2 + \Delta^4} \right).
$$

Since we have 
$
\Delta \leq O\left(\frac{n}{\log n}\right)^{1/4}
$ the last term can be upper bounded by $n^{-10}$.

Take $t_1, t_2, ..., t_k$ to be the size of the groups. We know that $k \leq n$; thus by a union bound over $k$, the probability that there is one group that deviates from its mean with a multiplicative factor of more than $\Delta^4$ is at most $n^{-9}$. Therefore, with probability at least $1 - n^{-9}$, all groups deviate from their mean with a multiplicative factor of less than $\Delta^4$. In this case, noting that $\sum_{i} t_i = n$, we get
$$
X \leq \Delta^4 \cdot \sum_{i} \frac{C\cdot t_i}{\Delta^5} = O(n/\Delta).
$$
Recalling that $X$ is the total number of online vertices that store their edges in $S$, we get that the total number of edges in $S$ is upper bounded by $\Delta \cdot O(n/\Delta) = O(n)$ w.h.p.

\paragraph{Case 2: $\Delta = \omega(\left(n/\log n\right)^{1/4})$.} Notice that the probability of storing any $u$ edges in $S$ is $n^{-5/4+o(1)}$. By union bounding over all $n$ online vertices, none will be stored in $S$ with a probability of $O(n^{-0.25+o(1)})$. 
\end{proof}
\subsection{Correctness and the Number of Colors}\label{sec:correctness}

In this section, we discuss why any two adjacent edges receive different colors in our algorithm. We can ignore the edges colored in Step~\ref{alg1_end} of \cref{alg:one} as they use a totally new set of colors from the rest of the edges. So we focus on the rest of the edges in the remainder of this section.

For an online vertex $u$, its edges by Step~\ref{alg1_findmatching} all receive different colors (and if this is not possible, we store $u$'s edges in $S$ which we discussed can be ignored earlier). So it remains to prove that there are no conflicts for the offline vertices.

For an offline vertex $v$, firstly, note that for any edge $e_i$:
\begin{flalign*}
    x_i^1 \in [0, c\Delta), \qquad
    x_i^2 \in [c\Delta, 2c\Delta), \qquad
    x_i^3 \in [2c\Delta, 3c\Delta).
\end{flalign*}

Since these ranges are disjoint, for two edges to receive the same color they must be in the same range.

Note that when the algorithm is running $\deg_v$ is increasing by one after we see an edge of $v$. Let $x_i^1 = x_j^1$ for two edges of $v$ where $e_i$ arrives before $e_j$. Define $\alpha$ to be the value of $\deg_v$ when the edge $e_i$ arrives and $\beta$ to be the value of $\deg_v$ when the edge $e_j$ arrives. Then we get that, 
$$
\big((r^1_{v} + \alpha) \bmod c \Delta\big) = \big((r^1_{v} + \beta) \bmod c \Delta\big),
$$
which is equivalent to $\alpha - \beta \bmod c\Delta = 0$. Since
$|\alpha - \beta| \leq \Delta$ and $c > 1$, this is a  contradiction and thus we cannot have $x_i^1 \neq x_j^1$. Same can be applied to the second and the third ranges of colors. Therefore, offline vertices also recieve distinct colors on all their edges and the coloring our algorithm finds is a valid edge coloring. 

Finally, it can be verified from \cref{alg:one} that the number of colors used is 
$3c\Delta + \Delta = O(\Delta).$
This section, put together with the previous section concludes the proof of \cref{lem:onesided} and, as discussed, this finishes the proof of \cref{thm:vertex-arrival}.

\subsection{Generalization to Batch Arrivals}\label{sec:batch}

Our \cref{alg:one} considers the vertex arrival model, i.e., all edges of an online vertex arrive at the same time. Here, we consider a more general {\em batch arrival} model that interpolates between edge arrivals and vertex arrivals. We show that our \cref{alg:one} can be generalized to this model. This generalization will be particularly useful in our algorithm of \cref{sec:edge-arrival}.

\begin{definition}[one-sided batch arrival model]\label{def:batch}
    In this model, we have a bipartite graph $G = (U, V, E)$ with  {\em online} vertices $U$ and {\em offline} vertices $V$. For a parameter $k$ of the problem, edges arrive in batches of size $k$ where all edges in the same batch are incident to the same online vertex $u$. The edges of a single online vertex $v$ can arrive in up to $\Delta / k$ different batches. Importantly, there is no guarantee that different batches of the same online vertex arrive consecutively (otherwise the model would be equivalent to the vertex arrival model). We use $B_i(u)$ to denote the $i$'th batch of online vertex $u$.
\end{definition}

We prove that \cref{alg:one} leads to the following bound in the one-sided batch arrival model. Note that setting $k = \Delta$ recovers \cref{thm:vertex-arrival}, so this is only more general.

\begin{lemma}\label{lem:batch-naive}
    There is a streaming algorithm that in the one-sided batch arrival model with batches of size $k$, uses $O(n)$ space and w.h.p. reports a proper $O(\Delta^2/k)$ edge coloring.
\end{lemma}
\begin{proof}
    Note that the only part of \cref{alg:one} where we use the assumption that all edges of the online vertices arrive at the same is in Step~\ref{alg1_findmatching}. There, the offline vertex of each edge $e_i$ of $u$ proposes three colors $\{x^1_i, x^2_i, x^3_i\}$ for $e_i$. Committing to choose a color from $\{x^1_i, x^2_i, x^3_i\}$ for each edge $e_i$ already ensures there will be no conflicts among the edges of the offline vertices. On the other hand, our \cref{lem:kout} guarantees that for each online vertex $u$, with probability $1-O(1/\Delta^5)$, there is a proper coloring of all of its edges using their proposed colors. However, to find this proper coloring, it is important to have all edges of the online vertex at once. If, as in our case in this lemma, the edges of the online vertex arrive in batches, then we cannot ensure that edges in two different batches of the same online vertex $u$ receive different colors. To fix this, first we only properly color the edges of each batch of the online vertex using their proposed colors. If this assigns color $c$ to edge $e$, the final color that we report for $e$ is $(c, i)$ where $i$ denotes which batch of $u$ this edge $e$ belongs to. Since there are at most $\Delta/k$ batches for each vertex, the total number of colors needed with this approach is $O(\Delta \cdot \Delta/k) = O(\Delta^2/k)$. Finally, note that the space remains $O(n)$ since the only additional information that we need is the number of batches that have arrived for each vertex which can be stored with $O(n)$ counters, one for each online vertex.
\end{proof}

\section{Edge-Arrivals: $O(\Delta)$ Edge-Coloring  with $\widetilde{O}(n\sqrt{\Delta})$ Space}\label{sec:edge-arrival-special}

In this section, we prove \cref{thm:main} for the special case where $s = \sqrt{\Delta}$. That is, we obtain an asymptotically optimal $O(\Delta)$ edge-coloring using $\widetilde{O}(n \sqrt{\Delta})$ space. We later prove \cref{thm:main} for the full range of $s$ in \cref{sec:edge-arrival}.

To solve the edge arrival model, we present a reduction to the one-sided batch arrival model of \cref{def:batch}. Throughout this section, we only consider the batch arrival model for batches of size $k := \lceil \sqrt{\Delta} \rceil$.

\begin{claim}\label{cl:edge-to-batch-4}
    Suppose that there is a streaming algorithm $\mc{A}$ that edge colors a bipartite graph of maximum degree $\Delta$ using $O(\Delta)$ colors under one-sided batch arrivals with batches of size $k = \Theta(\sqrt{\Delta})$ using space $\widetilde{O}(n \sqrt{\Delta})$. Then there is a streaming edge coloring $\mc{A}'$ algorithm that $O(\Delta)$ edge colors any general graph of maximum degree $\Delta$ under edge arrivals using $\widetilde{O}(n\sqrt{\Delta})$ space.
\end{claim}
\begin{proof}
Let us first assume that graph $G$ is bipartite with vertex sets $U$ and $V$ but its edges arrive in an arbitrary order.  We start with a buffer $T = \emptyset$ and add any edge that arrives in the stream to $T$. Whenever there is some vertex $u$ that has at least $k$ edges in $T$, we remove those edges of $u$ from $T$ and feed them into the batch-arrival algorithm $\mc{A}$. This way, each batch has exactly $k$ edges all adjacent to the same vertex. However, this vertex could belong to either of $U$ and $V$. That is, this is not one-sided arrivals but rather two-sided arrivals. The reduction to one-sided arrivals simply follows from \cref{cl:twosided-to-onesided} at the expense of multiplying the final number of colors by two. Finally, if the stream ends and there are any remaining edges in $T$, we color them all using $\Delta$ colors via offline algorithms. Note that since each vertex in $T$ has degree at most $k$ at any point, the space of this reduction is only $O(nk) = O(n\sqrt{\Delta})$ as desired. 

From the reduction above, we get that under edge arrivals, there is an algorithm $\mc{A}''$ that edge-colors the graph using $C(\Delta)$ colors for some constant $C$ and $\widetilde{O}(n\sqrt{\Delta})$ space provided that the graph is bipartite. We now apply \cref{lem:bipartization}. This gives an algorithm for general graphs under edge-arrivals that uses $\widetilde{O}(n\sqrt{\Delta})$ space and the number of colors that it uses is at most
$$
C\Delta + C\Delta/2 + C\Delta/4 + \ldots + C \leq 2C\Delta = O(\Delta).\qedhere
$$
\end{proof}

Our goal for the rest of this section is to prove the following algorithm which plugged to \cref{cl:edge-to-batch-4} proves \cref{thm:main} for $s = \sqrt{\Delta}$:

\begin{lemma}\label{lem:batch-alg-for-4}
    There is a streaming algorithm that edge-colors any bipartite graph of maximum degree $\Delta$ under one-sided batch-arrivals (\cref{def:batch}) with batches of size $k = \Theta(\sqrt{\Delta})$ using $O(\Delta)$ colors and $\widetilde{O}(n\sqrt{\Delta})$ space.
\end{lemma}
\begin{proof}
One may wonder whether we can now directly apply the naive algorithm of \cref{lem:batch-naive} in order to prove \cref{lem:batch-alg-for-4}. This is not doable because \cref{lem:batch-naive} would require $O(\Delta^2/k) = O(\Delta^{1.5})$ colors which is much larger than our desired $O(\Delta)$ colors. So we need more ideas.

First, we assume that $\Delta \geq 300\log^2 n$. If not, the whole graph has at most $O(n \Delta) = \widetilde{O}(n)$ edges, so we can store them all and run an offline $\Delta$ edge coloring algorithm.

Consider the subgraphs $H_1$, $H_2, \ldots, H_k$ where 
$$
H_i := \bigcup_{u \in U}^{} B_{i}(u),
$$ 
recalling that $B_i(u)$ is the $i$'th batch of online vertex $u$. In words, $H_i$ is the collection of the $i$'th batches of all online vertices.  Observe that the arrival order of the edges of each subgraph $H_i$ follow the vertex arrival model. That is, for each online vertex $v$, all of its edges arrive at the same time. Note also that the maximum degree on the online side of each subgraph $H_i$ is at most $k$. Now, if we also had the same upper bound of $k$ on the degrees on the offline side of each subgraph $H_i$, then we could run $\Delta / k = O(\sqrt{\Delta})$ instances of \cref{alg:one}, which works under vertex arrivals, for each of the $k$ subgraphs $H_i$ all in parallel, using disjoint color palettes. This way, the total number of colors used would be $\frac{\Delta}{k} \cdot O(k) = O(\Delta)$ and the space would be $\frac{\Delta}{k} \cdot O(n) = O(n\sqrt{\Delta})$. Unfortunately, however, the offline vertices in each $H_i$ may have degree as large as $\Delta$. This is because all edges of an offline vertex might be in the first batch of their corresponding online vertices, making the maximum degree in $H_1$ equal to $\Delta$ rather than $k$. This is the main obstacle that we overcome in the remainder the proof.

Let us for each online vertex $u$ choose a random number $b_u \in [k]$ at the beginning of the algorithm uniformly and independently at random. Now for the $i$'th batch of $u$ define its {\em permuted batch number} $\rho_u(i)$ to be
$$
\rho_u(i) := (i + b_u) \bmod k.
$$
Doing so, we now redefine $H_{i}$ for $i \in [k]$ as
$$
H_i := \bigcup_{u \in U} B_{\rho^{-1}_u(i)}(u).
$$
In words, $H_i$ is now the graph including, for each online vertex $u$, its edges in its $j$'th batch such that $\rho_u(j) = i$. Since the total number of batches of each online vertex $u$ is at most $\Delta/k \leq k$, all of $u$'s batches receive different batch numbers.  It remains to bound the degree of offline vertices in each $H_i$. The following lemma asserts that with high probability this is in fact $O(k)$.

\begin{lemma}\label{lem:offlinesrandomization}
    All $H_i$'s have maximum degree at most $2k$ with probability $1-1/\poly(n)$.
\end{lemma}

\begin{proof}
    Due to symmetry, let us focus on bounding the degree of a given offline vertex $v$ in $H_1$. Let neighbors of $v$ be $u_1, u_2, ..., u_l$ where each $u_i$ is an online vertex. Define $X_i$ to be the indicator random variable where $X_i = 1$ iff the edge $(v, u_i)$ receives a batch number of 1. The degree of $v$ in $H_1$ can be written as
    $$
    X = \sum_{i = 1}^l X_i.
    $$
    Since $l \leq \Delta$ and also $\E[X_i] = \frac{1}{k}$, by linearity of expectation, we get
    $$
    \E[X] = \sum_{i = 1}^l \E[X_i] \leq \frac{\Delta}{k} \leq k.
    $$
   
    Observe that all $b_{u_i}$'s are chosen at the beginning of the algorithm and uniformly at random. Therefore the batch number of any given edge being 1 has a probability of $\frac{1}{k}$. Since all $b_{u_i}$'s are independent of each other, also all $X_i$'s will be independent. Thus, we can apply the Chernoff bound for $\delta = 1$ and having $\Delta \geq 300 \log^2 n$ (as assumed at the start of the proof of \cref{lem:batch-alg-for-4}) we get, 
    $$
    \Pr[\text{degree } v \text{ in } H_1 \geq 2k] = \Pr\left[\sum_{i = 1}^l X_i \geq 2k \right] \leq \exp\left(\frac{-k}{3}\right) \leq n^{-10}
    $$
    By a union bound over all $n$ choices of $v$ and over $k$ choices of $H_i$, we get
    $$
    \Pr[\text{exists } v \in V \text{ and } i \in [k] \text{ such that degree } v \text{ in } H_i \geq 2k] \leq n^{-8}.\qedhere
    $$
\end{proof}

Now that both the offline and online vertices of each $H_i$ have degree at most $2k$, we can follow the approach outline before to achieve a $O(\Delta)$ coloring with $O(n\sqrt{\Delta})$ space. The final algorithm is formalized below as \cref{alg:two}.

\begin{figure}[h]
\begin{protocol}\label{alg:two}
    A streaming edge coloring for bipartite graphs under one-sided batch arrivals with batches of size $k = \Theta(\sqrt{\Delta})$. The algorithm uses $O(\Delta)$ colors and $O(n\sqrt{\Delta})$ space. This is the algorithm used for \cref{lem:batch-alg-for-4} which plugged into \cref{cl:edge-to-batch-4} proves \cref{thm:main} for $s = \sqrt{\Delta}$.

    \smallskip\smallskip
    \textbf{Parameter:} $k := \sqrt{\Delta}$, $c := 300$ (\cref{lem:offlinesrandomization})
    
    \begin{enumerate}[leftmargin=15pt]
        \item If $\Delta \leq c \cdot \log^2 n$ we read all the edges from the stream and color them with $\Delta$ colors, so assume for the rest of the algorithm that $\Delta > c \cdot \log^2 n$.
        \item As the edges of $G$ arrive, we decompose them into $k$ subgraphs $H_1, \ldots, H_k$ and run $k$ instances of the vertex-arrival \cref{alg:one} in parallel on these subgraphs using disjoint colors in these instances, and an overall space of $O(nk)$. Each graph $H_i$ will have maximum degree $k$, so overall we use $O(k \cdot k) = O(\Delta)$ colors.
        \item For each online vertex $u$ choose a random number $b_u \in [k]$ uniformly and independently. \label{alg4_3}
        \item For each online vertex $u$ store a counter $I_u$. This will keep track of the number of {\em batches} of $u$ seen at any point. Initially we have $I_u \gets 0$.
        \label{alg4_4}

        \item Upon arrival of a batch of size $k$ for an online vertex $u$:
        \begin{enumerate}
            \item $I_u \gets I_u + 1$.
            \item Add the edges of $u$ to graph $H_x$ where
            $$x := (I_u + b_u) \bmod k.$$
            \item These will be the only edges of $u$ in graph $H_x$, thus this is indeed a bipartite graph under one-sided vertex arrivals and so we can run the instance of \cref{alg:one} on graph $H_x$ to color the edges of $u$ in it.
        \end{enumerate}
        
    \end{enumerate}
\end{protocol}
\end{figure}

\paragraph{Space-complexity:} Since we are running $k = O(\sqrt{\Delta})$ instances of \cref{alg:one} in \cref{alg:two}, the space needed for keeping all these instances is also $O(n\sqrt{\Delta})$ as each has a space of $O(n)$ words by \cref{cl:space-highprob}.
Moreover, note that the size of $T$ never exceeds $nk$ or else there must be a vertex of degree at least $k$ in it. We also use at most $O(n)$ space for the counters and random numbers (specified in Step~\ref{alg4_3} and \ref{alg4_4}). So in total, the space that \cref{alg:two} uses is $O(nk) = O(n\sqrt{\Delta})$.

\paragraph{Number of colors:} For the number of colors used by \cref{alg:two}, note that we run $k$ instances of \cref{alg:one}, and in each instance the maximum degree is  w.h.p. $2k$ by \cref{lem:offlinesrandomization}. Since \cref{alg:one} uses linear colors in the maximum degree, the total number of colors is thus $O(k \cdot k) = O(\Delta)$ w.h.p.

This completes the proof of \cref{lem:batch-alg-for-4}.
\end{proof}

As discussed, plugging \cref{lem:batch-alg-for-4} into \cref{cl:edge-to-batch-4} completes the proof of \cref{thm:main} for $s = \sqrt{\Delta}$, concluding this section.

\section{Edge-Arrivals: $O(\Delta^{1.5}/s + \Delta)$ Edge-Coloring  with $\widetilde{O}(ns)$ Space}\label{sec:edge-arrival}

In \cref{sec:edge-arrival-special}, we proved \cref{thm:main} for the special case where $s = \sqrt{\Delta}$. In this section, we prove \cref{thm:main}, more generally, for the full range of $s$. Based on the color/space trade-off of \cref{thm:main}, we can assume w.l.o.g. that $1 \leq s \leq \sqrt{\Delta}$. To see this, note that increasing $s$ increases the space usage $\widetilde{O}(ns)$ of the algorithm but in return reduces the number of colors $O(\Delta^{1.5}/s + \Delta)$ that it uses. However, when $s > \sqrt{\Delta}$, the second term in  $(\Delta^{1.5}/s + \Delta)$ becomes larger than the first, and so there is no point in increasing $s$ (and thus the space) any further. With this assumption of $1 \leq s \leq \sqrt{\Delta}$, we show that the graph can be edge colored via $O(\Delta^{1.5}/s)$ colors using $\widetilde{O}(ns)$ space.

Similar to \cref{sec:edge-arrival-special}, we start with a reduction to batch-arrivals. However, the reduction here is slightly different and more general.

\begin{claim}\label{cl:edge-to-batch}
    Suppose that there is a streaming algorithm $\mc{A}$ that edge colors a bipartite graph of maximum degree $\Delta$ using $O(\Delta^{1.5}/s)$ colors under one-sided batch arrivals with batches of size $k = \Theta(\sqrt{\Delta})$ using space $\widetilde{O}(n s)$. Then there is a streaming edge coloring $\mc{A}'$ algorithm that $O(\Delta^{1.5}/s)$ edge colors any general graph of maximum degree $\Delta$ under edge arrivals using $\widetilde{O}(ns)$ space.
\end{claim}
\begin{proof}
    Let us first assume that graph $G$ is bipartite with vertex sets $U$ and $V$ but its edges arrive in an arbitrary order. 
    Similar to the proof of \cref{cl:edge-to-batch-4} section, we buffer the edges in a set $T$. However, instead of allowing $T$ to have size $nk$ for $k := \lceil \sqrt{\Delta} \rceil$ which is too large when $s \ll k$, we ensure that $|T| \leq n s$ at any point. To do this, whenever $|T|$ reaches $ns$, we consider two cases: 

\paragraph{Case 1: Each vertex has less than $k$ edges in $T$.} In this case, we color using an offline algorithm all of the edges in $T$ using $k$ fresh colors and delete them from memory. We claim that in total, this only requires $O(\Delta^{1.5}/s)$ colors. To see this, note that every time that this event occurs, we remove $ns$ edges of the graph from $T$. Thus, in total, we repeat this process at most $m/ns$ times where $m$ is the number of edges in the original graph. Since $m \leq n \Delta$,  this is upper bounded by $n\Delta/ns = \Delta/s$. Each time we use $k$ fresh colors, so in total the number of colors used is $O(k \cdot \Delta/s) = O(\Delta^{1.5}/s)$.

\paragraph{Case 2: There is a vertex $u$ with $k$ edges in $T$.} In this case, we feed in the $k$ edges of $u$ as a batch to the batch-arrival algorithm $\mc{A}$. Notice that $u$ may belong to both $U$ and $V$. However, using the reduction of \cref{cl:twosided-to-onesided}, we can convert this into one-sided batch arrivals at the expense of multiplying the number of colors by two.

From the reduction above, we get that under edge arrivals, there is an algorithm $\mc{A}''$ that edge-colors the graph using $f(\Delta) = C(\Delta^{1.5}/s)$ colors for some constant $C$ and $\widetilde{O}(n s)$ space provided that the graph is bipartite. We now apply \cref{lem:bipartization}. This gives an algorithm for general graphs under edge-arrivals that uses $\widetilde{O}(ns)$ space and the number of colors that it uses is at most
$$
f(\Delta) + f(\Delta/2) + f(\Delta/4) + \ldots + f(1).
$$
Since $s \leq \sqrt{\Delta}$, $f(\Delta)$ is at least linear in $\Delta$ and so this sum is upper bounded by $O(f(\Delta)) = O(\Delta^{1.5}/s)$. 
\end{proof}

Our goal for the rest of this section is to prove the following algorithm which plugged to \cref{cl:edge-to-batch} proves \cref{thm:main}.

\begin{lemma}\label{lem:batch-alg}
    For any $1 \leq s \leq \sqrt{\Delta}$, there is a streaming algorithm that edge-colors any bipartite graph of maximum degree $\Delta$ under one-sided batch-arrivals (\cref{def:batch}) with batches of size $k = \Theta(\sqrt{\Delta})$ using $O(\Delta^{1.5}/s)$ colors and $\widetilde{O}(n s)$ space.
\end{lemma}

\begin{proof} 
We assume that $\Delta \geq 900\log^2 n$. Otherwise we can store the whole graph and $\Delta$ edge color it using $\widetilde{O}(n)$ space.

Since edges arrive in batches of size $k$, each vertex will have at most $\Delta/k$ batches. A naive generalization of the algorithm of \cref{sec:edge-arrival-special} for \cref{lem:batch-alg-for-4} would run $\Delta/k = O(\sqrt{\Delta})$ parallel instances of vertex arrival algorithms. However, since each instance would require $O(n)$ space, this aggregates to $O(n\sqrt{\Delta})$ overall, which is larger than the desired space of $O(ns)$ when $s \ll \sqrt{\Delta}$. 

To get around the issue mentioned above, we partition the batches of each online vertex $u$ into {\em groups}. In particular, we call the first $\frac{k}{s}$ batches of $u$ group 1, the next $\frac{k}{s}$ batches group 2, and so on. Hence, the total number of groups $u$ has will be at most $\frac{k}{k/s} = s$. Observe that the degree of an online vertex $u$ in each group $i$ is at most
\begin{equation}\label{eq:deg-online-group}
\text{(\# batches in group $i$) } \cdot \text{(degree in each batch)} =  \frac{k}{s} \cdot k = O\left(\frac{\Delta}{s}\right).
\end{equation}
Similar to the challenge we faced in the previous section, there is no upper bound smaller than $\Delta$ on the degrees of the offline vertices in each group. To get around this, for each vertex $u$ we draw a random number $g_u \in [s]$ at the beginning of the stream uniformly and independently. Now, for the $i$'th group, define its {\em permuted group number} as
\begin{equation}\label{eq:groupnum}
    \sigma_u(i) = (i + g_u) \bmod s,
\end{equation}
and let $H_{i, u}$ include the edges of $u$ in its $j$'th group where $\sigma_u(j) = i$. Now, for any $i \in [s]$, we define 
$$
H_i := \bigcup_{u \in U} H_{i,u}
$$
to be the graph including all the edges with permuted group numbers $i$.

The following claim is similar to \cref{lem:offlinesrandomization} of the previous section and bounds the maximum degrees in each $H_i$ provided that $\Delta$ is sufficiently large.

\begin{claim}\label{clm:degree_group}
With probability $1-1/\poly(n)$, it happens for all $i \in [s]$ that the maximum degree in $H_i$ is at most $O(\Delta/s)$.
\end{claim}

\begin{proof}
The proof is similar to \cref{lem:offlinesrandomization} of \cref{sec:edge-arrival-special}. Take an online vertex $u \in H_i$. Let $j$ be the group for which $\sigma_u(j) = i$, noting that there is only one group $j$ satisfying this (as $\sigma_u$ is a permutation of $[s]$ to $[s]$). Since, by \cref{eq:deg-online-group} vertex $u$ has $O(\Delta/s)$ edges in group $j$, it will also have $O(\Delta/s)$ edges in graph $H_i$. Note that this guarantee is deterministic for online vertices. For the offline vertices, we use the Chernoff bound. In particular, take an edge $e=(u, v)$ of an offline vertex $v$. This edge is added to graph $H_i$ if and only if the group $j$ that $e$ belongs to gets group number $\sigma_u(j) = i$. Since $\sigma_u$ is a random permutation of $[s]$, this happens with probability $1/s$. Thus, in expectation, at most $\Delta/s$ edges of $u$ belong to $H_i$. Moreover, note this choice is independent for different edges of $v$ since the random permutations of the online vertices are independent. Hence, we can apply the Chernoff bound and get that
$$
\Pr[v \text{ has degree } \geq 2\Delta/s \text{ in $H_i$}] \leq \exp(- \frac{\Delta}{3s}).
$$
Noting that $s \leq \sqrt{\Delta}$ and $\Delta \geq 900 \log^2 n$, we get $\Delta/3s \geq \sqrt{\Delta}/3 \geq 10 \log n$. Hence, the probability above is upper bounded by $n^{-10}$. A union bound over all $n$ choices of offline vertices and all the $s \leq n$ groups, we get that with probability $1-1/n^8$, the maximum degree in each $H_i$ is indeed $O(\Delta/s)$.
\end{proof}

We run $s$ parallel instances of the algorithm of \cref{lem:batch-naive}, one for each \underline{group}, each instance using its own colors. The space requirement of each of these instances is $\widetilde{O}(n)$, so this uses $\widetilde{O}(ns)$ space which is the desired bound. For the number of colors, recall that the edges of each group come in $k/s$ batches, so \cref{lem:batch-naive} requires $O(\Delta' \cdot k / s)$ colors where $\Delta'$ is the maximum degree in the group, which by \cref{clm:degree_group} is upper bounded by $O(\Delta/s)$. So in total, the number of colors that the algorithm uses is at most 
\begin{flalign*}
\underbrace{s}_\text{\# of groups} \times \underbrace{k/s}_\text{batches in each group } \times \underbrace{\Delta/s}_\text{max degree in each group} 
 = O(\Delta k/s) = O(\Delta^{1.5}/s),
\end{flalign*}
as desired. This completes the proof; see \cref{alg:three} for the pseudo-code of the algorithm.
\end{proof}

\begin{figure}
\begin{protocol}\label{alg:three}
    Streaming edge coloring for bipartite graphs under one-sided batch arrivals with batches of size $k = \Theta(\sqrt{\Delta})$. The algorithm uses $O(\Delta^{1.5}/s)$ colors and $O(n s)$ space. This is the algorithm used for \cref{lem:batch-alg} which plugged into \cref{cl:edge-to-batch} proves \cref{thm:main}.

    \smallskip\smallskip
    \textbf{Parameter:} $k := \sqrt{\Delta}$, $s$ (space we use = $O(ns)$)
    
    \begin{enumerate}[leftmargin=15pt]
        \item If $\Delta \leq 900 \log^2 n$ store the whole graph and color it with $\Delta$ colors using offline algorithms. So assume for the rest of the algorithm that $\Delta > 900 \log^2 n$.
    
        \item We run $s$ parallel instances of the  algorithm of \cref{lem:batch-naive} on subgraphs $H_1, \ldots, H_s$ of the input graph $G$. Each subgraph $H_i$ is a bipartite graph with one vertex part including its {\em online} vertices and the other part including {\em offline} vertices. Each $H_i$ will have maximum degree $O(\Delta/s)$ (by \cref{clm:degree_group}), and its edges will arrive in batches of size $k$ where all edges in the same batch are adjacent to the same online vertex. Since each instance requires $\widetilde{O}(n)$ space by \cref{lem:batch-naive}, the total space needed is $\widetilde{O}(ns)$.
        
        \item For each online vertex $u$, pick and store a random number $g_s \in [s]$ uniformly and independently.
        \item For each online vertex $u$ store a counter $I_u$. This will keep track of the number of {\em batches} of $u$ seen at any point. Initially we have $I_u \gets 0$.
    
    \item Upon arrival of a batch of size $k$ for an online vertex $u$:
    \label{alg3_readedge}
        \begin{enumerate}
            \item $I_u \gets I_u + 1$.
            \item Let $I'_u := \lceil I_u/(k/s)\rceil$ be the {\em group number} of $u$. In particular, the first $k/s$ batches of $u$ will be in group 1, the second $k/s$ batches will be in group 2 and so on so forth.
            \item Take $k$ edges of $u$ as a batch, remove them from the buffer $T$, and feed them as a batch to graph $H_i$ where
            $$
            i := (I'_u + g_u) \bmod s.
            $$
            \item Using the instance of \cref{lem:batch-naive} run on graph $H_i$, we color the edges of $u$.
        \end{enumerate}
    \end{enumerate}
\end{protocol}
\end{figure}

% \clearpage

\appendix
\section{Appendix: A Reduction To Bipartite Graphs}\label{apx:bipartization}
\label{appendix:a}

In this section, we present the proof of \cref{lem:bipartization} using a random bipartization idea that has been used extensively in the literature of edge colorings. To our knowledge, this idea was first used by \citet*{DBLP:journals/jal/KarloffS87} in the context of parallel algorithms and by \citet{BehnezhadDHKS19} in the context of streaming edge coloring. We emphasize that we claim no novelty in the algorithm of this section and merely present it here for the sake of self-containment.

\begin{proof}[Proof of \cref{lem:bipartization}]
    If $\Delta \leq 10 \log n$, then we store the whole graph in memory using $\widetilde{O}(n)$ space and edge color the graph via $(\Delta + 1)$ colors using offline algorithms. So assume that $\Delta > 10\log n$.

    Let $V$ be the vertex-set of the graph. We partition $V$ into two subsets $A$ and $B$ by putting each vertex randomly in either $A$ or $B$ chosen uniformly and independently from the rest of the vertices. Now consider the bipartite subgraph $G'$ of $G$ including any edge $(u, v)$ of $G$ that has one endpoint in $A$ and one endpoint in $B$. Note that conditioned on $v \in A$, each neighbor of $v$ in the original graph belongs to $B$ independently with probability 1/2. This means that the expected degree of $v$ in $G'$ is $\deg_G(v)/2$. By a simple application of the Chernoff bound and recalling that $\Delta > 10\log n$, we get that $G'$ has maximum degree $\Delta/2 + O(\sqrt{\Delta\log n}) = (1/2+o(1))\Delta$ with probability $1-1/n^{10}$. We can thus run the bipartite edge coloring algorithm on $G'$. Notice, however, that this leaves the edges that go from $A$ to $A$ or from $B$ to $B$ uncolored. We recursively apply the same algorithm. That is, we recursively partition $A$ and $B$ into two subsets each. This results in a subgraph of maximum degree $\Delta/4 + O(\sqrt{\Delta \log n})$ w.h.p. We continue this for $O(\log \Delta)$ steps until the resulting subgraph has maximum degree smaller than $O(\log n)$, at which point we store the whole subgraph in memory and color it using offline algorithms. Overall, the total number of colors used by the algorithm is
    $$
    \sum_{i=1}^{\log \Delta} f\left(\left(\frac{1}{2^i}+o(1)\right)\Delta\right) \leq \sum_{i=1}^{\log \Delta} f\left(\frac{\Delta}{2^{i-1}}\right).
    $$
    Note that this partitioning of the vertices can be done at the beginning of the stream before any edges arrive. It only suffices to store, for each vertex, which of the two random subsets it belongs to at each of the $O(\log \Delta)$ levels. So this only requires an overhead space of $O(n\log \Delta)$. Additionally, if the edges of the original graph arrive under vertex arrivals, then so do the edges of the random bipartite graphs.
\end{proof}

\bibliographystyle{plainnat}
\bibliography{references}

\end{document}